\documentclass[unsortedaddress,superscriptaddress,pra,notitlepage]{revtex4}

\usepackage{amsfonts}
\usepackage[latin2]{inputenc}
\usepackage{t1enc}
\usepackage{amssymb}
\usepackage{amsmath}
\usepackage[mathscr]{eucal}
\usepackage{amsthm}
\usepackage{array}
\usepackage{color}
\usepackage[hidelinks]{hyperref}
\usepackage{braket}
\usepackage{scalerel}
	
\usepackage{graphicx}
\usepackage{xcolor}
\usepackage{tikz}

\usetikzlibrary{calc,decorations.pathreplacing}
\usetikzlibrary{arrows,shapes}

\theoremstyle{definition} 
\newtheorem{defin}{Definition}[section]

\newtheorem{thm}[defin]{Theorem}

\newtheorem{rem}[defin]{Remark}

\newtheorem{ex}[defin]{Example}
\newtheorem{cor}[defin]{Corollary}

\newtheorem{lemma}[defin]{Lemma}

\renewcommand{\theHWexercise}{\theHWexercise$^*$}

\def\theta{\vartheta}
\def\hil{{\mathcal H}}

\def\A{{\mathcal A}}

\def\B{{\mathcal B}}

\def\C{{\mathcal C}}

\def\F{{\mathcal F}}

\def\I{{\mathcal I}}
\def\J{{\mathcal J}}

\def\S{{\mathcal S}}
\def\T{{\mathcal T}}

\def\half{\frac{1}{2}}
\def\iff{\Longleftrightarrow}
\def\imp{\Longrightarrow}

\def\bN{\mathbb{N}}
\def\bC{\mathbb{C}}

\def\bR{\mathbb{R}}
\def\bZ{\mathbb{Z}}

\def\bT{\mathbb{T}}
\def\bz{\left(}
\def\jz{\right)}

\def\egy{\mathbf 1}

\def\what{\widehat}

\def\rho{\varrho}

\def\test{\mathrm{test}}

\def\oll{\overline}

\def\valt{\cdot}

\def\n{\langle n \rangle}
\def\transl{\mathrm{trans}}
\def\lep{\mu}
\def\uep{\nu}
\def\hh{\chi}

\newcommand{\ki}[1]{\textit{\textit{#1}}}

\newcommand{\s}{\mbox{ }}
\newcommand{\ds}{\mbox{ }\mbox{ }}
\newcommand{\fv}{\hat}
\newcommand{\norm}[1]{\left\| #1\right\|}

\newcommand{\inner}[2]{\left\langle #1 , #2\right\rangle}

\newcommand{\abs}[1]{\left| #1 \right|}

\newcommand{\diad}[2]{\left|#1\right\rangle\!\left\langle #2\right|}

\newcommand{\pr}[1]{\diad{#1}{#1}}

\newcommand{\floor}[1]{\left\lfloor #1\right\rfloor}

\newcommand{\fock}[1]{\Gamma(#1)}
\newcommand{\fockop}[1]{#1_{F}}
\newcommand{\tests}[1]{\mathbb{T}\bz #1\jz}

\renewcommand{\phi}{\varphi}
\renewcommand{\theta}{\vartheta}
\renewcommand{\epsilon}{\varepsilon}

\newcommand{\dd}{\mathrm{d}}

\newcommand{\integers}{\mathbb{Z}}
\newcommand{\hilb}{\mathcal{H}}

\newcommand{\bounded}[1]{\mathcal{B}\left(#1\right)}
\newcommand{\car}[1]{\mathrm{CAR}\left(#1\right)}

\newcommand{\anticomm}[2]{\left\{#1,#2\right\}}

\makeatletter
\renewcommand{\p@enumii}{}
\makeatother

\DeclareMathOperator{\Tr}{Tr}

\DeclareMathOperator{\supp}{supp}

\DeclareMathOperator{\ran}{ran}

\DeclareMathOperator{\spann}{span}

\DeclareMathOperator{\ootimes}{\otimes\ldots\otimes}

\DeclareMathOperator{\conv}{\star}

\DeclareMathOperator{\divv}{\Delta}

\DeclareMathOperator{\Ft}{\F}
\DeclareMathOperator*{\medotimes}{\scalerel*{\otimes}{\textstyle\sum}}
\DeclareMathOperator*{\medoplus}{\scalerel*{\oplus}{\textstyle\sum}}

\begin{document}

\title{Super-exponential distinguishability of correlated quantum states}

\author{Gergely Bunth}
\email{gbunthy@gmail.com}

\author{G\'abor Mar\'oti}
\email{marotigabor1995@gmail.com}

\author{Mil\'an Mosonyi}
\email{milan.mosonyi@gmail.com}

\affiliation{
MTA-BME ``Lend\"ulet'' Quantum Information Theory Research Group
}

\affiliation{
Department of Analysis, Institute of Mathematics,\\
 Budapest University of Technology and Economics, M\H uegyetem rkp.~3., H-1111 Budapest, Hungary
}

\author{Zolt\'an Zimbor\'as}
\email{zimboras.zoltan@wigner.hu}

\affiliation{
MTA-BME ``Lend\"ulet'' Quantum Information Theory Research Group
}

\affiliation{Wigner Research Centre for Physics, H-1525, P.O.Box 49, Budapest, Hungary}

\begin{abstract}
\centerline{\textbf{Abstract}}
\vspace{.3cm}

In the problem of asymptotic binary i.i.d.~state discrimination,
the optimal asymptotics of the type I and the type II error probabilities is in general an exponential decrease to 
zero as a function of the number of samples; the set of achievable exponent pairs is 
characterized by the quantum Hoeffding bound theorem. 
A super-exponential decrease for both types of error probabilities is 
only possible in the trivial case when the two states are orthogonal, and hence 
can be perfectly distinguished using only a single copy of the system. 

In this paper we show that a qualitatively different behaviour 
can occur when there is correlation between the samples. 
Namely, we use gauge-invariant and translation-invariant quasi-free states
on the algebra of the canonical anti-commutation relations
to exhibit pairs of states on an infinite spin chain with the 
properties that a) all finite-size restrictions of the states have invertible density operators, and b) the type I and the type II error probabilities both decrease to zero 
at least with the speed 
$e^{-nc\log n}$ with some positive constant $c$, i.e., with a super-exponential speed
in the sample size $n$.
Particular examples of such states include the ground states of the $XX$ model corresponding to different transverse magnetic fields.
In fact, we prove our result in the setting of binary composite hypothesis testing, and hence it can be applied to prove super-exponential distinguishability of the hypotheses that the 
transverse magnetic field is above a certain threshold vs.~that it is below a strictly lower
value. 
\end{abstract}

\maketitle

\section{Introduction}

In the problem of simple binary state discrimination, an experimenter is presented with a 
quantum system that is either in some state $\omega^{(0)}$ or in another state 
$\omega^{(1)}$.
The experimenter's task is to guess which one the true state of the system is, based on 
measurements on the system. It is easy to see that even the most elaborate measurement 
and classical post-processing scheme cannot outperform single $2$-outcome (binary) 
measurements when the goal is to minimize the probability of an erroneous decision. 
More precisely, there are two types of error probabilities to consider:
erroneously identifying the state as $\omega^{(1)}$ (type I error), or 
erroneously identifying the state as $\omega^{(0)}$ (type II error), 
and the goal is to minimize some combination of the two. 
It is easy to see that (in the finite-dimensional case, at least), 
perfect discrimination (i.e., when both error probabilities are zero) is possible if and only 
if the density operators of the two states have orthogonal supports. 

The error probabilities can be reduced if the experimenter has access to multiple identical 
copies of the system, and in the asymptotic analysis of the problem one is interested in 
the achievable asymptotic behaviours of the two error probabilities along all possible 
sequences of binary measurements (tests) as the number of copies tends to infinity. 
In general, the best achievable asymptotics is an exponential decrease to zero for both error 
probabilities; the set of the achievable exponent pairs is described by the 
quantum Hoeffding bound theorem \cite{ANSzV,Hayashicq,Nagaoka}. Faster
(super-exponential) decrease is possible if and only if the supports of the states are different. 
For instance, if $\supp\omega^{(1)}\not\subseteq\supp\omega^{(0)}$ then there exists a test sequence 
along which the type I error is constant zero (hence its exponent is $+\infty$), while the type II error 
decreases exponentially fast (with the exponent being the R\'enyi zero-divergence of 
$\omega^{(0)}$ and $\omega^{(1)}$). A faster than exponential decrease for both error 
probabilities is possible if and only if the supports are orthogonal, in which case 
both errors can be made zero trivially for any finite number of copies. 

The above are well-known in the i.i.d.~(independent and identically distributed) case, i.e., 
when all the samples are prepared in the same state, and there is no correlation between the 
different samples. Correlated scenarios can be conveniently described using the concept of 
the $C^*$-algebra of an infinite spin chain, $\C_{\bZ}=\otimes_{k\in\bZ}\B(\hil)$, 
where $\hil$ is a 
finite-dimensional Hilbert space describing a single system.
In this case the candidate states $\omega^{(0)}$ and $\omega^{(1)}$ can be described by positive 
linear functionals on $\C_{\bZ}$ that take $1$ on the identity; their restrictions
to any sub-algebra $\otimes_{k\in\Lambda}\B(\hil)$ corresponding to a 
finite subset $\Lambda$ of samples (equivalently, a finite part of the chain)
can be described by density operators in the usual way. A state on the infinite chain
is translation-invariant if the density operator
of any finite subsystem $\Lambda$ is the same as that of any of its translates; in particular, 
the single-site density operators are all the same (i.e., 
the outcomes of the same measurement performed at different sites are
identically distributed). 
In this picture, a measurement on $n$ consecutive samples is described by 
a measurement on a length $n$ part of the chain, and the asymptotics is 
studied in the setting where this length is allowed to go to infinity. 
Obviously, error exponents are more difficult to determine in the correlated scenario, but
rather general results are available in the setting of Stein's lemma, where 
one of the errors is not required to decrease exponentially
\cite{Sanov_corr}, and in the setting of the Hoeffding bound
for thermal states of translation-invariant finite-range Hamiltonians,
and more generally, for states that satisfy a certain factorization property
\cite{HMO2}. In these cases, however, the entropic quantities
(Umegaki- and quantum R\'enyi relative entropies)
characterizing the achievable exponent pairs are given by regularized formulas, 
and cannot be explicitly computed in general. 

A particular class of correlated states where explicit formulas are available 
can be obtained from translation-invariant and gauge-invariant quasi-free states
on the algebra of canonical anti-commutation relations (CAR algebra).
Such a state
is specified by a measurable function on $[0,2\pi)$ with values in $[0,1]$,
called the symbol of the state; see Section \ref{sec:prelim} below for details.
The achievable exponent pairs were determined for a pair of such states
in \cite{MHOF}, with explicit expressions for the relevant entropic expressions, in the case 
where the symbols of the two states, denoted by $\fv q$ and $\fv r$,
are bounded away from $0$ and $1$ in the sense that 
$\eta\le\fv q(x),\fv r(x)\le 1-\eta$ for all $x\in[0,2\pi)$ for some
$\eta>0$.
In this case the regularized quantum R\'enyi $\alpha$-divergences of the two states 
are finite for every $\alpha>0$, and the best achievable asymptotics is an exponential decay for both error probabilities.

Our main contribution in this paper is showing that for certain pairs of quasi-free states,
super-exponential discrimination is possible. More precisely, we show that 
if the symbols $\fv q$ and $\fv r$ are such that 
there exists a non-degenerate interval on which $\fv q$ is constant $0$ and 
$\fv r$ is constant $1$ then there exists a sequence of tests along which both error 
probabilities decrease at least with the speed $e^{-nc\log n}$,
where $n$ is the sample size.
In the same time, unless $\fv q$ is constant zero 
and $\fv r$ is constant $1$ (up to sets of measure zero), then all the local densities of both states
are invertible, and hence it is not only impossible to make both error probabilities vanish 
for a finite sample size, but if one of the error probabilities is made zero then the other is necessarily equal to $1$. 
This is very different from what can be seen in the i.i.d.~case, and to the best of our knowledge, this is the first time that such a behaviour is presented in the literature. 
%

The structure of the paper is as follows.
In Section \ref{sec:CAR} we review the necessary basics about quasi-free states on the CAR 
algebra. In Section \ref{sec:state disc} we explain the notions of error exponents and 
super-exponential distinguishability for translation-invariant states on the spin chain 
and on the CAR algebra. In Section \ref{sec:superexp} we prove our main result described above. In fact, we state and prove a more general result in the framework of composite
state discrimination, showing super-exponential 
distinguishability of two sets of quasi-free states with invertible local density operators.
In Section \ref{sec:ort} we give various characterizations of super-exponential 
distinguishability of states in terms of regularized divergences.

\section{Preliminaries}
\label{sec:prelim}

\subsection{Quasi-free states on the CAR algebra}
\label{sec:CAR}

Here we summarize the necessary basics about quasi-free states on the CAR algebra. For more details and proofs we refer to \cite{BR2,Fannes-CAR,AlickiFannes-book,dierckx2008quasifree,ZZK2014}.

For a complex Hilbert space $\hil$, we will denote the set of bounded operators on 
$\hil$ by $\B(\hil)$, and use the notation $\tests{\hil}:=\{T\in\B(\hil):\,0\le T\le I\}$ for 
the set of \ki{tests} on $\hil$.

For vectors $\phi_1,\ldots,\phi_k$ in a complex Hilbert space $\hil$, let 
\begin{align*}
\phi_1\wedge\ldots\wedge\phi_k:=\frac{1}{\sqrt{k!}}\sum_{\sigma\in \mathfrak{S}_k}\varepsilon(\sigma)\phi_{\sigma(1)}\ootimes\phi_{\sigma(k)}
\end{align*}
denote their anti-symmetrized tensor product, where
$\mathfrak{S}_k$ stands for the set of permutations of $k$ elements and $\varepsilon(\sigma)$ 
for the sign of the permutation $\sigma$. For any $k\in\bN\setminus\{0\}$,
the $k$-th anti-symmetric tensor power of $\hil$ is
\begin{align*}
\hil^{\wedge k}:=\oll{\spann\{\phi_1\wedge\ldots\wedge\phi_k:\,\phi_i\in\hil,\,i=1,\ldots,k\}},
\end{align*}
where the overline denotes the closure in operator norm, 
and we define $\hil^{\wedge 0}:=\bC$.
The Hilbert space of a fermionic system with single-particle Hilbert space $\hil$ is
the \ki{anti-symmetric Fock space}
\begin{align*}
\fock{\hil}:=\medoplus_{k\in\bN}\hil^{\wedge k},
\end{align*}
where $\hil^{\wedge k}=\{0\}$ for every $k>\dim\hil$. 
For an operator $A\in\bounded{\hil_1,\hil_2}$, let
$A^{\wedge k}:= {\left.A^{\otimes k}\right|}_{\hil_1^{\wedge k}}$ as an operator
on $\hil_1^{\wedge k}$, and 
\begin{align*}
\fockop{A} &:= \medoplus_{k\in\bN}A^{\wedge k},
\end{align*}
with $A^{\wedge 0}:=A^{\otimes 0}:=1\in\B(\bC)$. Clearly, $\fockop{A}$ is bounded if and only 
if $\dim\hil<+\infty$ or $\norm{A}\le 1$. 
If $V:\,\hil_1\to\hil_2$ is an isometry/unitary 
then $\fockop{V}$ is an isometry/unitary from $\fock{\hil_1}$ into/onto
$\fock{\hil_2}$ with the property $\fockop{V}\hil_1^{\wedge k}\subseteq\hil_2^{\wedge k}$.

For each $\phi\in\hilb$, the corresponding \emph{creation operator} 
$c(\phi)$ is the unique bounded linear extension of the map
\begin{align*}
\phi_1\wedge\ldots\wedge\phi_k\mapsto\phi\wedge\phi_1\wedge\ldots\wedge\phi_k,
\ds\ds\ds\phi_1,\ldots,\phi_k\in\hil,
\end{align*}
and the corresponding \emph{annihilation operator} is its adjoint, 
$a(\phi):=c(\phi)^*$. These operators satisfy the \emph{canonical anti-commutation relations (CARs)},
\begin{align} \label{eqn:cars}
\anticomm{a(\phi)}{a(\psi)}= 0, \ds\ds\ds 
\anticomm{a(\phi)}{a^*(\psi)}= \inner{\phi}{\psi}I,\ds\ds \phi,\psi\in\hil.
\end{align}
The C*-algebra generated by $\{a(\phi):\,\phi\in\hil\}$ 
is called the \emph{algebra of the canonical anti-commutation relations} (or \emph{CAR-algebra}) corresponding to the single-particle Hilbert space $\hil$, and is denoted by $\car{\hil}$.
Note that $\phi\mapsto c(\phi)$ is complex linear and $\phi\mapsto a(\phi)$ is complex 
anti-linear. Thus, if $\hil$ is separable and $(e_i)_{i=1}^{\dim\hil}$ is an 
orthonormal basis (ONB) in it then 
$\car{\hil}$ is the closure of the linear span of the identity and all the 
multinomials of the form $a(e_{i_1})^*\ldots a(e_{i_n})^*a(e_{j_m})\ldots a(e_{j_1})$,
$i_1<\ldots<i_n$, $j_1<\ldots<j_m$. For any isometry/unitary $V:\,\hil_1\to\hil_2$,
$\fockop{V}(\valt)\fockop{V}^*$ is a homomorphism/isomorphism from $\car{\hil_1}$ to $\car{\hil_2}$ with the property $\fockop{V}a(\phi)\fockop{V}^*=a(V\phi)$, $\phi\in\hil_1$.
The \ki{even part} $\car{\hil}_+$ of $\car{\hil}$ is the sub-algebra left invariant by the 
\ki{parity automorphism} $\fockop{(-I)}(\valt)\fockop{(-I)}$. This is exactly the closure of the linear span 
of all multinomials with an even number of terms. 

If $\hil$ is finite-dimensional and $e_1,\ldots,e_d$ is an orthonormal basis in $\hil$ then 
$\{e_{i_1}\wedge\ldots\wedge e_{i_k}:\,1\le i_1<\ldots<i_k\le d\}$ is
an ONB in $\hil^{\wedge k}$; in particular, 
$\dim\hil^{\wedge k}=\binom{d}{k}$ and $\dim\F(\hil)=2^{\dim\hil}$. 
Let $\ket{0}=\begin{bmatrix} 1 \\ 0 \end{bmatrix}$, 
$\ket{1}=\begin{bmatrix} 0 \\ 1 \end{bmatrix}$ be the canonical ONB of $\bC^2$. Then 
\begin{align}\label{JW1}
U_e:\,e_{i_1}\wedge\ldots\wedge e_{i_k}\mapsto \otimes_{j=1}^d\ket{x_j},\ds\ds
x_j:=\begin{cases}1,&j\in\{i_1,\ldots,i_k\},\\ 0,&\text{otherwise},\end{cases}
\end{align}  
is a unitary from $\fock{\hil}$ to $(\bC^2)^{\otimes d}$, and   
\begin{align*}
U_ea(e_j)^*U_e^*=\underbrace{\sigma_z\ootimes\sigma_z}_{j-1\text{ times}}\otimes\begin{bmatrix}0 & 0\\ 1 & 0\end{bmatrix}\otimes\underbrace{I\ootimes I}_{d-j\text{ times}}
\end{align*}
is easy to verify.
The map $U_e(\valt)U_e^*:\,\B(\F(\hil))=\car{\hil}\to\B((\bC^2)^{\otimes d})=\B(\bC^2)^{\otimes d}$ is called the \ki{Jordan-Wigner isomorphism} corresponding to the given ONB.
The \ki{particle number operator} is 
\begin{align*}
N_{\hil}:=\medoplus_{k=0}^{d}kI_{\hil^{\wedge k}}=
\sum_{i=1}^{d}a(e_i)^*a(e_i)
=
U_e^*\bz\sum_{i=1}^d
\underbrace{I\ootimes I}_{i-1\text{ times}}\otimes\begin{bmatrix}0 & 0\\ 0 & 1\end{bmatrix}\otimes\underbrace{I\ootimes I}_{d-i\text{ times}}\jz
U_e.
\end{align*}
The eigen-values of $N_{\hil}$ are $0,\ldots,d$, with spectral projections
\begin{align}\label{numop spectral}
P^{N_{\hil}}_k=
\underbrace{0\oplus\ldots\oplus 0}_{k\text{ times}}\oplus I_{\hil^{\wedge k}}\oplus \underbrace{0\oplus\ldots\oplus 0}_{d-k\text{ times}}
=
U_e^*\bz\sum_{\Lambda\subseteq[d],\,|\Lambda|=k}
\bz\medotimes_{i\in\Lambda}\begin{bmatrix}0 & 0\\ 0 & 1\end{bmatrix}\jz
\medotimes
\bz\medotimes_{i\in[d]\setminus\Lambda}\begin{bmatrix}1 & 0\\ 0 & 0\end{bmatrix}\jz\jz 
U_e,
\end{align}
where $[d]:=\{1,\ldots,d\}$.
Note that $N_{\hil}$ is defined in a basis-independent way, and the equalities above are valid for any ONB.

A \ki{state} on $\car{\hil}$ is a positive linear functional that takes the value $1$ on $I$.    
For any positive semi-definite (PSD) operator $Q\in\B(\hil)$ with $Q\le I$ there exists a unique state $\omega_Q$ on $\car{\hil}$ 
(called the \ki{gauge-invariant quasi-free state with symbol $Q$})
with the property
\begin{align} \label{eqn:quasifree_functional}
\omega_Q\left(a(\phi_1)^*\ldots a(\phi_n)^*a(\psi_m)\ldots a(\psi_1)\right) 
= \delta_{mn}\det\left\{\Braket{\psi_i|Q\phi_j}\right\}_{i,j=1}^n.
\end{align}
It is easy to verify that when $\hil$ is finite-dimensional, the density operator 
$\what\omega_Q$ of $\omega_Q$ can be explicitly given as
\begin{align}\label{quasifree density}
\what\omega_Q=\prod_{j=1}^d\bz q_ja(e_j)^*a(e_j)+(1-q_j)a(e_j)a(e_j)^*\jz
=
U_e^*\bz\medotimes_{j=1}^d\begin{bmatrix}1-q_j & 0 \\ 0 & q_j\end{bmatrix}\jz U_e,
\end{align}
where $Q=\sum_{j=1}^d q_j\pr{e_j}$ is any eigen-decomposition of $Q$, and
$U_e$ is the unitary corresponding to the ONB $(e_j)_{j=1}^d$ as in \eqref{JW1}.
Note that for all $1\le i_1<\ldots<i_k\le d$, $e_{i_1}\wedge\ldots\wedge e_{i_k}$
is an eigen-vector of $\what\omega_Q$ with eigen-value
$\bz\prod_{j\in\{i_1,\ldots,i_k\}}q_j\jz\cdot\bz \prod_{j\in[d]\setminus\{i_1,\ldots,i_k\}}(1-q_j)\jz$. 
This implies immediately that if $1$ is not an eigen-value of $Q$ then 
$\what\omega_Q$ can be written as 
\begin{align*}
\what\omega_Q=\det(I-Q)\fockop{\bz\frac{Q}{I-Q}\jz}.
\end{align*} 
Quasi-free states emerge as equilibrium states of non-interacting fermionic systems. For instance, if the single-particle Hamiltonian $H$ of a system of non-interacting fermions 
is such that $e^{-\beta H}$ is trace-class 
then 
the Gibbs state of the system at inverse temperature $\beta$ is the quasi-free state with symbol
$Q=\frac{e^{-\beta H}}{I+e^{-\beta H}}$ (see, e.g., \cite[Proposition 5.2.23]{BR2}).

Consider now a fermionic chain with a single mode at each site. The
single-particle Hilbert space of this system is $\hil=\ell^2(\integers)$, the standard basis 
of which we denote by $\{\egy_{\{k\}}:\,k\in\integers\}$. The \emph{translation operator} is 
the unitary $U^{\transl}=\sum_{k\in\bZ}\diad{\egy_{\{k+1\}}}{\egy_{\{k\}}}$, and $\tau(\valt):=\fockop{U^{\transl}}(\valt)\fockop{(U^{\transl})}^*$ gives an automorphism of 
$\car{\ell^2(\integers)}$ with the property  
$\tau(a(\phi))=a(U^{\transl}\phi)$, $\phi\in\hil$. 
A quasi-free state $\omega_Q$ is called \emph{translation-invariant} if 
$\omega_Q\circ\tau=\omega_Q$, which is easily seen to be equivalent to 
$U^{\transl}Q=QU^{\transl}$, i.e., the translation-invariance of the symbol $Q$. For instance, in the above example a translation-invariant single-particle Hamiltonian $H$ yields a translation-invariant quasi-free state as the equilibrium state of the system. 
Translation-invariant operators on $\ell^2(\bZ)$ commute with each other and they are simultaneously diagonalized by the Fourier transformation
\begin{align*}
\Ft:\,\ell^2(\bZ)\to L^2([0,2\pi)),\ds 
\Ft\egy_{\{k\}}:=\chi_{k},\ds 
\chi_{k}(x):=\frac{e^{ikx}}{\sqrt{2\pi}},\ds x\in[0,2\pi),\ds k\in\bZ.
   \end{align*}
That is,
every translation-invariant operator $A$ arises in the form $A=\Ft^*M_{\fv a}\Ft$, where $M_{\fv a}$ denotes the multiplication operator by a bounded measurable function $\fv a$ on $[0,2\pi)$. As a consequence, the matrix entries of translation-invariant operators 
in the canonical ONB
are constants along diagonals; more explicitly, 
for any translation-invariant operator $A\in\B(\ell^2(\bZ))$,
\begin{align}\label{Toeplitz}
A_{k,j}:=\inner{\egy_{\{k\}}}{A\egy_{\{j\}}}=
\frac{1}{2\pi}\int_{[0,2\pi)}e^{-i(k-j)x}\fv a(x)\,\dd x
,\ds\ds\ds k,j\in\bZ.
\end{align}

A measurement on a subsystem corresponding to modes at the sites $\n:=\{0,\ldots,n-1\}$ has
measurement operators in the 
$C^*$-subalgebra 
$\A_n\subseteq\car{\ell^2(\bZ)}$ generated by $\{a(\phi):\,\phi\in\hil_n\}$,
\begin{align*}
\hil_n:=\spann\{\egy_{\{k\}}:\,k\in\n\}\subseteq \ell^2(\bZ).
\end{align*}
This subalgebra is naturally isomorphic to $\car{\bC^{\n}}$.
It is easy to see that if the 
state of the infinite chain is given by a quasi-free state with symbol $Q$ then 
the statistics of any such local measurement is given by the quasi-free state 
$\omega_{Q_n}$ on $\car{\bC^{\n}}$
with symbol $Q_n:=V_n^*QV_n$, where 
$V_n$ is the natural embedding of $\bC^{\n}$ into $\ell^2(\bZ)$. 

\begin{lemma} \label{lem:subsys_eigen}
Let $\fv a:[0,2\pi)\to [0,+\infty)$ be a non-negative bounded measurable function, and let $A=\Ft^*M_{\fv a}\Ft$ be the corresponding translation-invariant operator on $\ell^2(\bZ)$. 
The following are equivalent:
\begin{enumerate}
\item\label{1 ev 1}
$0$ is an eigen-value of $V_n^*AV_n$ for some $n\in\bN$; 
\item\label{1 ev 2}
$0$ is an eigen-value of $V_n^*AV_n$ for every $n\in\bN$; 
\item\label{1 ev 3}
$A=0$;
\item\label{1 ev 4}
$\fv a$ is equal to $0$ almost everywhere.
\end{enumerate}
\end{lemma}
\begin{proof} 
The equivalence \ref{1 ev 4}$\iff$\ref{1 ev 3} is obvious, as are the implications
\ref{1 ev 3}$\imp$\ref{1 ev 2}$\imp$\ref{1 ev 1}, and hence we only need to prove 
\ref{1 ev 1}$\imp$\ref{1 ev 4}. 
Assume therefore that $V_n^*AV_n\psi=0$ for some $\psi\in\bC^{\n}\setminus\{0\}$. Then 
\begin{align*}
0=\inner{\psi}{V_n^*AV_n\psi}=\norm{A^{1/2}V_n\psi}^2
=
\norm{\Ft A^{1/2}\Ft^*\Ft V_n\psi}^2=
\norm{M_{\fv a^{1/2}}\Ft V_n\psi}^2,
\end{align*}
whence $\fv a^{1/2}\Ft V_n\psi=0$ almost everywhere. Since $\Ft V_n\psi$ is a non-zero 
trigonometric polynomial that can only have finitely many zeros, this implies that 
$\fv a$ is $0$ almost everywhere.
\end{proof}

\begin{cor}\label{cor:invertible density}
Let $Q=\Ft^*M_{\fv q}\Ft\in\B(\ell^2(\bZ))$ be the symbol of a translation-invariant quasi-free 
state. If $\fv q$ is neither almost everywhere zero 
nor almost everywhere $1$ then for every  $n\in\bN$,
$\what\omega_{Q_n}$ is an invertible density operator on $\fock{\bC^{\n}}$.
\end{cor}
\begin{proof}
Applying Lemma \ref{lem:subsys_eigen} to $\fv a:=\fv q$ yields that
$0$ is not an eigen-value of $Q_n$ for any $n\in\bN$. 
Applying Lemma \ref{lem:subsys_eigen} to $\fv a:=1-\fv q$ yields that 
$1$ is not an eigen-value of $Q_n$, either, for any $n\in\bN$. 
Thus, the assertion follows from \eqref{quasifree density}.
\end{proof}
\medskip

Finally, a symbol $Q$ on $\bC^{\n}$ is translation-invariant (or rotation-invariant), if it 
commutes with the $n$-dimensional translation unitary 
$U^{\transl}_n=\sum_{k=0}^{n-1}\diad{\egy_{\{k+1\}}}{\egy_{\{k\}}}$, where the addition is modulo $n$. Such operators are also called \ki{circular}, and are simultaneously diagonalized 
by the \ki{$n$-dimensional discrete Fourier transformation}
\begin{align}\label{DFT}
\Ft_n:\,\bC^{\n}\to\bC^{\n}, \ds\ds\ds
\Ft_n\egy_{\{k\}}:=\frac{1}{\sqrt{n}}\sum_{j=0}^{n-1}e^{i\frac{2\pi}{n}kj}\egy_{\{j\}},\ds\ds\ds
k\in\n.
\end{align}
That is, $U^{\transl}_nQ=QU^{\transl}_n$, $0\le Q\le I$, if and only if $Q=\Ft_n^*M_{\fv q}\Ft_n$
for some $\fv q\in[0,1]^{\n}$.

\subsection{Asymptotic binary state discrimination}
\label{sec:state disc}

The infinite spin chain algebra with single-site finite-dimensional Hilbert space $\hil$ 
is defined as
\begin{align*}
\C_{\bZ}(\hil):=\otimes_{k\in\bZ}\B(\hil):=
\Bigg(\bigcup_{\Lambda\subseteq\bZ\,\mathrm{finite}}\C_{\Lambda}(\hil)\Bigg)\Bigg/_{
\large{\sim}}
\ds\ds,
\end{align*}
where $\C_{\Lambda}(\hil):=\otimes_{i\in\Lambda}\B(\hil)$, and 
every $A\in\C_{\Lambda}(\hil)$ is naturally identified with 
$A\otimes\bz\otimes_{i\in\Lambda'\setminus\Lambda}I\jz
\in\C_{\Lambda'}(\hil)$ for $\Lambda'\supseteq\Lambda$.
A translation-invariant state $\omega$ on the infinite spin chain is specified by 
density operators $\omega_{\Lambda}$ in $\C_{\Lambda}(\hil)$ such that 
$\Tr_{\Lambda'\setminus\Lambda}\omega_{\Lambda'}=\omega_{\Lambda}$ and
$\omega_{\Lambda+k}=\omega_{\Lambda}$ for any finite $\Lambda\subseteq\bZ$ and $k\in\bZ$.
Equivalently, $\omega$ is a positive linear functional on the $C^*$-algebra $\C_{\bZ}(\hil)$,
with $\omega(I)=1$, such that for the translation automorphism $\tau$ we have $\omega\circ\tau=\omega$,
and the $\omega_{\Lambda}$ are the density operators of its restrictions onto
$\C_{\Lambda}(\hil)$.

Given two sets translation-invariant states
$\Omega^{(0)}=\{\omega^{(0,i)}\}_{i\in\I}$ and $\Omega^{(1)}=\{\omega^{(1,j)}\}_{j\in\J}$, 
a state discrimination protocol of sample size $n$
to decide if the true state of the system belongs to $\Omega^{(0)}$
(null-hypothesis $H_0$) or to $\Omega^{(1)}$
(alternative hypothesis $H_1$), 
is specified by a \ki{test}
$T_n\in\C_{[1,n]}(\hil)$ with $0\le T_n\le I$, representing a measurement with outcomes $0$ and 
$1$, with corresponding measurement 
operators $T_n$ and $I-T_n$, respectively. If the outcome of the measurement is $k$, the 
experimenter accepts hypothesis $H_k$ to be true. 
The (worst-case) \ki{type I error probability} of incorrectly rejecting $H_0$, and the 
\ki{type II error probability} of incorrectly accepting $H_0$, respectively, are given by 
\begin{align*}
\alpha_n(T_n):=\sup_{i\in\I}\Tr\omega^{(0,i)}_{[1,n]} (I-T_n),\ds\ds\ds
\beta_n(T_n):=\sup_{j\in\J}\Tr\omega^{(1,j)}_{[1,n]} T_n.
\end{align*}
A test $T_n$ is \ki{projective}, if $T_n^2=T_n$. 
Given a sequence of tests $\vec{T}=(T_n)_{n\in\bN}$, with $T_n\in\C_{[1,n]}(\hil)$, $n\in\bN$, the corresponding 
\ki{type I and type II error exponents} are defined, respectively, as
\begin{align}\label{errexp}
\alpha^{\exp}(\vec{T}):=\liminf_{n\to+\infty}-\frac{1}{n}\log\alpha_n(T_n),\ds\ds\ds
\beta^{\exp}(\vec{T}):=\liminf_{n\to+\infty}-\frac{1}{n}\log\beta_n(T_n).
\end{align}
We say that $\Omega^{(0)}$ and $\Omega^{(1)}$ can be 
\ki{super-exponentially distinguished}, if there exists a test sequence 
$\vec{T}$ along which $\alpha^{\exp}(\vec{T})=+\infty=\beta^{\exp}(\vec{T})$.

As it was shown in \cite{ArakiXY} (see also \cite[Section 5.3]{MosonyiPhD} for a detailed exposition) every 
translation-invariant gauge-invariant quasi-free state $\omega$ on 
$\car{\ell^2(\bZ)}$ can be mapped into 
a translation-invariant state $\tilde\omega$ on the spin chain 
$\C_{\bZ}(\bC^2)$ with the preservation of the locality structure. In particular, 
given two sets $\Omega^{(0)}=\{\omega^{(0,i)}\}_{i\in\I}$ and 
$\Omega^{(1)}=\{\omega^{(1,j)}\}_{j\in\J}$ of such states on $\car{\ell^2(\bZ)}$, and numbers
$\alpha,\beta\in[0,1]$, there exists a (projective) test 
$T_n\in\C_{[1,n]}(\bC^2)$ such that 
$\sup_{i\in\I}\Tr\tilde\omega^{(0,i)}_{[1,n]}(I-T_n)=\alpha$, 
$\sup_{j\in\J}\Tr\tilde\omega^{(1,j)}_{[1,n]}T_n=\beta$, 
if and only if there exists a (projective) test
$S_n\in\car{\hil_n}$ such that 
$\sup_{i\in\I}\omega^{(0,i)}(I-S_n)=\alpha$, 
$\sup_{j\in\J}\omega^{(1,j)}(S_n)=\beta$. 
Hence, in order to explore the achievable error exponent pairs
for the pair $\tilde\Omega^{(0)}=\{\tilde\omega^{(0,i)}\}_{i\in\I}$,
$\tilde\Omega^{(1)}=\{\omega^{(1,j)}\}_{j\in\J}$,
one can work directly on the CAR algebra with 
$\Omega^{(0)}$ and $\Omega^{(1)}$. Thus, we introduce the following:

\begin{defin}\label{def:superexp}
Let $\{\fv q_i\}_{i\in\I}$ and $\{\fv r_j\}_{j\in\J}$ be measurable functions from 
$[0,2\pi)$ to $[0,1]$, defining the translation-invariant quasi-free states 
$\Omega_Q:=\{\omega_{Q^{(i)}}\}_{i\in\I}$, $\Omega_R:=\{\omega_{R^{(j)}}\}_{j\in\J}$ on $\car{\ell^2(\bZ)}$. 
We say that 
$\Omega_Q$ and $\Omega_R$ can be super-exponentially distinguished, if 
there exists a sequence $T_n\in\car{\hil_n}$, $n\in\bN$, such that 
\begin{align*}
\liminf_{n\to+\infty}-\frac{1}{n}\log\sup_{i\in\I}\omega_{Q^{(i)}}(I-T_n)
=+\infty=
\liminf_{n\to+\infty}-\frac{1}{n}\log\sup_{j\in\J}\omega_{R^{(j)}}(T_n).
\end{align*}
\end{defin}

By the above, the sets of states $\Omega_Q$ and $\Omega_R$ on the CAR algebra
are super-exponentially distinguishable
if and only if so are the sets of states 
$\tilde\Omega_Q$ and $\tilde\Omega_R$ on the spin chain.

\section{Super-exponential distinguishability}
\label{sec:superexp}

In this section we prove the main result of the paper:

\begin{thm}\label{thm:superexp}
Let $\{\fv q_i\}_{i\in\I}$ and $\{\fv r_j\}_{j\in\J}$ be measurable functions from 
$[0,2\pi)$ to $[0,1]$, defining the translation-invariant quasi-free states 
$\Omega_Q:=\{\omega_{Q^{(i)}}\}_{i\in\I}$, $\Omega_R:=\{\omega_{R^{(j)}}\}_{j\in\J}$ on $\car{\ell^2(\bZ)}$. 
If there exists an interval $[\lep,\uep]\subseteq[0,2\pi)$ of positive length such that 
$\fv q_i$ is constant $0$ and $\fv r_j$ is constant $1$ on 
it for every $i\in\I$ and $j\in\J$, then
$\Omega_Q$ and $\Omega_R$ are super-exponentially 
distinguishable. 

If, moreover, $\fv q_i$ is not almost everywhere $0$ and $\fv r_j$ is not 
almost everywhere $1$ on $[0,2\pi)$ for every $i\in\I$ and $j\in\J$, 
then their local densities $\what\omega_{Q_n^{(i)}}$ and 
$\what\omega_{R_n^{(j)}}$, $n\in\bN$, are all invertible.
\end{thm}

In fact, the above theorem follows immediately from a more detailed statement given in 
Theorem \ref{thm:superexp2} below, which we prove in several steps. 

The main intuition behind the proof is the following.
(For simplicity we take $|\I|=|\J|=1$, $Q^{(1)}=:Q$, $R^{(1)}=:R$.)
Although the symbols $Q,R\in\B(\ell^2(\bZ))$ commute with each other, this is not true anymore 
for their restrictions $Q_n$ and $R_n$ onto $\bC^{\n}$, unless $Q$ or $R$ is a constant multiple of the 
identity. On the other hand, if instead of restrictions of translation-invariant symbols onto
finite-dimensional subspaces we considered translation-invariant symbols 
$Q_n,R_n$ on 
the single-particle Hilbert space $\bC^{\n}$ of a length $n$ finite chain
(with periodic boundary conditions, or equivalently, rotation-invariant symbols on a finite ring) then any two such symbols would commute with each other, and would be simultaneously
diagonalized by the discrete Fourier transformation; see the end of Section \ref{sec:CAR}.
Now, the analogous condition to the one in Theorem \ref{thm:superexp} in the 
finite-dimensional case would be that the functions
$\fv q_n,\fv r_n\in\bC^{\n}$ satisfy $\fv q_n(k)=0=1-\fv r_n(k)$, $k=l+1,\ldots,l+m$ for some 
$l,m\in\bC^{\n}$ (modulo $n$). Hence, for the projection
$E_n:=\F_n^*\sum_{k=l+1}^{l+m}\pr{\egy_{\{k\}}}\F_n$, we would have
$\Tr E_nQ_n=0=\Tr E_n(I-R_n)$. 
A key technical ingredient of our proof, given in Lemma \ref{lemma:main} and 
Corollary \ref{cor:error bounds} below,
is that for any such projection, one can construct a test, using the spectral decomposition 
of the particle number operator on the subspace $\ran E_n$, such that 
the type I and type II error probabilities are upper bounded by a simple expression 
involving only $\Tr E_n$, $\Tr E_nQ_n$ and $\Tr E_n(I-R_n)$; in particular, if the latter two are $0$ then so are the error probabilities.
When $Q_n$ and $R_n$ are the non-commuting restrictions of $Q,R\in\B(\ell^2(\bZ))$, we can 
still follow the above strategy, where instead of making the upper bounds exactly 
zero, we can make them sufficiently small, as shown in 
Lemmas \ref{lemma:DFT diagonal}, \ref{lemma:constant fv} and \ref{lemma:final bounds}.

\begin{lemma}\label{lemma:main}
Let $\hil$ be a finite-dimensional Hilbert space, and 
\begin{align}\label{numop test}
S_{\hil}:=\sum_{k=0}^{\floor{\dim\hil/2}}P^{N_{\hil}}_k.
\end{align}
For any $A\in\B(\hil)$ with $0\le A\le I$, 
\begin{align*}
\omega_A(I-S)\le\bz\frac{8\Tr A}{\dim\hil}\jz^{\frac{\dim\hil}{2}},\ds\ds\ds\ds\ds
\omega_A(S)\le\bz\frac{8\Tr (I-A)}{\dim\hil}\jz^{\frac{\dim\hil}{2}}.
\end{align*}
\end{lemma}
\begin{proof}
Let $d:=\dim\hil$, $S:=S_{\hil}$, and 
\begin{align*}
A=\sum_{i=1}^{d}a_i\pr{e_i},
\end{align*}
be an eigen-decomposition of $A$.
By \eqref{numop spectral},
\begin{align}
S&=
\sum_{k=0}^{\floor{d/2}}U_e^*\bz\sum_{\Lambda\subseteq[d],\,|\Lambda|=k}
\bz\medotimes_{j\in\Lambda}\begin{bmatrix}0 & 0\\ 0 & 1\end{bmatrix}\jz
\medotimes
\bz\medotimes_{j\in[d]\setminus\Lambda}\begin{bmatrix}1 & 0\\ 0 & 0\end{bmatrix}\jz\jz 
U_e.\label{numop JW1}
\end{align}
By \eqref{quasifree density},
\begin{align}
\what\omega_A
&=
U_e^*\bz\medotimes_{j=1}^d\begin{bmatrix}1-a_j & 0 \\ 0 & a_j\end{bmatrix}\jz U_e
\le
\begin{cases}
U_e^*\bz\medotimes_{j=1}^d\begin{bmatrix}1 & 0 \\ 0 & a_j\end{bmatrix}\jz U_e,\\
\s\\
U_e^*\bz\medotimes_{j=1}^d\begin{bmatrix}1-a_j & 0 \\ 0 & 1\end{bmatrix}\jz U_e.
\end{cases}
\label{density upper1}
\end{align}
Using \eqref{numop JW1} and the first bound in \eqref{density upper1}, we get
\begin{align*}
\omega_A(I-S)&=\Tr\what\omega_A(I-S)\\
&\le
\Tr\sum_{k=\floor{d/2}+1}^{d}
\sum_{\Lambda\subseteq[d],\,|\Lambda|=k}
\bz\medotimes_{j\in\Lambda}\begin{bmatrix}0 & 0\\ 0 & a_j\end{bmatrix}\jz
\medotimes
\bz\medotimes_{j\in[d]\setminus\Lambda}\begin{bmatrix}1 & 0\\ 0 & 0\end{bmatrix}\jz\\
&=
\sum_{k=\floor{d/2}+1}^d
\sum_{\Lambda\subseteq[d],\,|\Lambda|=k}\prod_{j\in\Lambda}a_j
\le
\sum_{\floor{d/2}+1}^d
\sum_{\Lambda\subseteq[d],\,|\Lambda|=k}\bz\frac{\sum_{j\in\Lambda}a_j}{k}\jz^{k},
\end{align*}
where the second inequality follows from the geometric-arithmetic mean inequality.
Using that 
\begin{align*}
\bz\frac{\sum_{j\in\Lambda}a_j}{k}\jz^{k}
\le
\bz\frac{\sum_{j\in\Lambda}a_j}{k}\jz^{d/2}
\le
\bz\frac{\sum_{j\in\Lambda}a_j}{d/2}\jz^{d/2}
\le
\bz\frac{\Tr A}{d/2}\jz^{d/2}
\end{align*}
for every $k\ge \floor{d/2}+1$,
we get 
\begin{align*}
\omega_A(I-S)\le
\bz\frac{\Tr A}{d/2}\jz^{d/2}\sum_{k=\floor{d/2}+1}^d
\underbrace{\sum_{\Lambda\subseteq[d],\,|\Lambda|=k}1}_{=\binom{d}{k}}
\le
2^{d}\bz\frac{\Tr A}{d/2}\jz^{d/2}=\bz\frac{8\Tr A}{d}\jz^{d/2}.
\end{align*}
Similarly, using 
\eqref{numop JW1} and the second bound in \eqref{density upper1} yields
\begin{align*}
\omega_A(S)&=\Tr\what\omega_AS\\
&\le
\Tr\sum_{k=0}^{\floor{d/2}}
\sum_{\Lambda\subseteq[d],\,|\Lambda|=k}
\bz\medotimes_{j\in\Lambda}\begin{bmatrix}0 & 0\\ 0 & 1\end{bmatrix}\jz
\medotimes
\bz\medotimes_{j\in[d]\setminus\Lambda}\begin{bmatrix}1-a_j & 0\\ 0 & 0\end{bmatrix}\jz\\
&=
\sum_{k=0}^{\floor{d/2}}
\sum_{\Lambda\subseteq[d],\,|\Lambda|=k}\prod_{j\in[d]\setminus\Lambda}(1-a_j)
\le
\sum_{k=0}^{\floor{d/2}}
\sum_{\Lambda\subseteq[d],\,|\Lambda|=k}\bz\frac{\sum_{j\in[d]\setminus\Lambda}(1-a_j)}{d-k}\jz^{d-k}\,.
\end{align*}
Using that 
\begin{align*}
\bz\frac{\sum_{j\in[d]\setminus\Lambda}(1-a_j)}{d-k}\jz^{d-k}
\le
\bz\frac{\sum_{j\in[d]\setminus\Lambda}(1-a_j)}{d-k}\jz^{d-d/2}
\le
\bz\frac{\sum_{j\in[d]\setminus\Lambda}(1-a_j)}{d-d/2}\jz^{d/2}
\le
\bz\frac{\Tr(I-A)}{d-d/2}\jz^{d/2}
\end{align*}
for all $k\le\floor{d/2}$,
we get 
\begin{align*}
\omega_A(S)\le
\bz\frac{\Tr(I-A)}{d/2}\jz^{d/2}
\sum_{k=0}^{\floor{d/2}}
\underbrace{\sum_{\Lambda\subseteq[d],\,|\Lambda|=k}1}_{=\binom{d}{k}}
\le
2^{d}\bz\frac{\Tr(I-A)}{d/2}\jz^{d/2}=\bz\frac{8\Tr(I-A)}{d}\jz^{d/2}.
\end{align*}
\end{proof}

\begin{cor}\label{cor:error bounds}
Let $\hil$ be a Hilbert space.
For every non-zero finite-rank projection $E$ on $\hil$ there exists 
an even projection $T\in\spann\{a(\phi):\,\phi\in\ran E\}$ such that 
for every $A\in\B(\hil)$, $0\le A\le I$,
\begin{align}\label{error bounds}
\omega_A(I-T)\le\bz\frac{8\Tr EA}{\Tr E}\jz^{\frac{\Tr E}{2}},\ds\ds\ds
\omega_A(T)\le\bz\frac{8\Tr E(I-A)}{\Tr E}\jz^{\frac{\Tr E}{2}}.
\end{align}
\end{cor}
\begin{proof}
Let $V_E$ be the identical embedding of $\ran E$ into $\hil$. 
For any even projection $S\in\car{\ran E}$,
$T:=\fockop{(V_E)}S\fockop{(V_E)}^*$ is an even projection in $\spann\{a(\phi):\,\phi\in\ran E\}$, and 
$\omega_A(I-T)=\omega_{A_E}(I-S)$, $\omega_A(T)=\omega_{A_E}(S)$, where
$A_E:=V_E^*AV_E$,
Noting that $\Tr A_E=\Tr AE$, $\Tr (I_{\ran E}-A_E)=\Tr E(I-A)$,
the assertion follows from 
Lemma \ref{lemma:main} by choosing $S:=S_{\ran E}$ as in \eqref{numop test}.
\end{proof}

\begin{cor} \label{cor:error bounds2}
Let $\{Q^{(i)}\}_{i\in\I},\{R^{(j)}\}_{j\in\J}\subseteq\bounded{\ell^2(\integers)}$ be symbols of quasi-free states
on $\car{\ell^2(\integers)}$. Assume that there exists a sequence of non-zero projections 
$E_n$ on $\bC^{\n}$, $n\in\bN$, such that 
\begin{align*}
\liminf_{n\to+\infty}\frac{\Tr E_n}{n}\log\frac{\Tr E_n}{\sup_{i\in\I}\Tr E_nQ^{(i)}_n}=+\infty,\ds\ds\ds
\liminf_{n\to+\infty}\frac{\Tr E_n}{n}\log\frac{\Tr E_n}{\sup_{j\in\J}\Tr E_n(I-R^{(j)}_n)}=+\infty.
\end{align*}
Then $\Omega_Q$ and $\Omega_R$ can be super-exponentially distinguished by 
even projective tests.
\end{cor}
\begin{proof}
Let $T_n$ be the projection corresponding to $V_nE_nV_n^*$ as in Corollary \ref{cor:error bounds}, where $V_n$ is the canonical embedding of 
$\bC^{\n}$ into $\ell^2(\bZ)$. Since $\Tr E_n=\Tr V_nE_nV_n^*$,
$\Tr E_nQ_n^{(i)}=\Tr V_nE_nV_n^*Q^{(i)}$,
$\Tr E_n(I-R_n^{(j)})=\Tr V_nE_nV_n^*(I-R^{(j)})$,
\eqref{error bounds} yields 
\begin{align}
-\frac{1}{n}\log\omega_{Q^{(i)}}(I-T_n)
&\ge
\frac{\Tr E_n}{2n}\log\frac{\Tr E_n}{8\Tr E_nQ_n^{(i)}},\label{nonas lower bound1}\\
-\frac{1}{n}\log\omega_{R^{(j)}}(T_n)
&\ge
\frac{\Tr E_n}{2n}\log\frac{\Tr E_n}{8\Tr E_n(I-R_n^{(j)})}\,.
\label{nonas lower bound2}
\end{align}
The statement follows by taking the infima over the respective index sets, and then the
liminf in $n$ in the above inequalities,
and noting that $0\le \frac{\Tr E_n}{2n}\log 8\le \half\log 8$.
\end{proof}
\medskip

Hence, in order to complete the proof of Theorem \ref{thm:superexp}, it is sufficient to show
that if $\Omega_Q$ and $\Omega_R$ are as in Theorem \ref{thm:superexp} then a sequence of projections as in 
Corollary \ref{cor:error bounds2} exists. 
For this, we will need some simple facts about Fourier transforms; see, e.g., \cite{stein2003fourier}
for details. 

In particular, 
recall that the $n$-th partial sum of the Fourier series of an integrable function on 
$[0,2\pi)$ is given by 
\begin{align*}
(S_nf)(x):=\sum_{k=-n}^ne^{ikx}\frac{1}{2\pi}\int_{[0,2\pi)}e^{-ikt}f(t)\,\dd t
=
(f\conv D_n)(x),
\end{align*}
where $D_n(x):=\frac{1}{2\pi}\sum_{k=-n}^n e^{ikx}=
\frac{1}{2\pi}\frac{\sin((n+1/2)x)}{\sin (x/2)}$ is the 
\ki{Dirichlet kernel}, and $\conv$ stands for the convolution. 
The $n$-th C\'esaro mean of the partial sums is 
\begin{align*}
(\what S_n f)(x)
&:=
\frac{1}{n}\sum_{k=0}^{n-1}S_n f(x)
=
\sum_{k=-n+1}^{n-1}\frac{n-|k|}{n}e^{ikx}
 \frac{1}{2\pi}\int_{[0,2\pi)}e^{-ikt}f(t)\,\dd t
=(f\conv F_n)(x),
\end{align*}
where $F_n(x):=\frac{1}{n}\sum_{k=0}^{n-1}D_n(x)=\frac{1}{2\pi n}\frac{\sin^2(nx/2)}{\sin^2(x/2)}$ is the \ki{Fej\'er kernel}. 
	
The following may be known; however, as we have not found a reference in the literature, we provide a detailed proof. Recall that $V_n$ is the canonical embedding of 
$\bC^{\n}$ into $\ell^2(\bZ)$. 

\begin{lemma}\label{lemma:DFT diagonal}
Let $\fv a$ be a bounded measurable complex-valued function on $[0,2\pi)$ and $A=\Ft^* M_{\fv a}\Ft$. Then the diagonal matrix entries of $\Ft_n V_n^*AV_n\Ft_n^*$ are given by
\begin{align*}
(\Ft_n V_n^*AV_n\Ft_n^*)_{k,k}=(\what S_n\fv a)\bz\frac{2\pi k}{n}\jz,\ds\ds\ds
k\in\n.
\end{align*}
\end{lemma}	
\begin{proof}
Let $A_n:=V_n^*AV_n$. By \eqref{Toeplitz} and \eqref{DFT}, 
\begin{align*}
(\Ft_n V_n^*AV_n\Ft_n^*)_{k,k}&=
\sum_{j,l=0}^{n-1}(\Ft_n)_{k,j}(V_n^*AV_n)_{j,l}(\Ft_n^*)_{l,k}\\
&=\frac{1}{n}\sum_{j,l=0}^{n-1}e^{i\frac{2\pi}{n}kj}e^{-i\frac{2\pi}{n}kl}
\frac{1}{2\pi}\int_{[0,2\pi)}e^{-i(j-l)x}\fv a(x)\,\dd x\\
&=
\sum_{m=-n+1}^{n-1}\frac{n-|m|}{n}e^{i\frac{2\pi}{n}km}
\frac{1}{2\pi}\int_{[0,2\pi)}e^{-imx}\fv a(x)\,\dd x\\
&=(\what S_n\fv a)\bz\frac{2\pi k}{n}\jz,
\end{align*}
where in the third equality above we replaced the summation over $j,l$ with a single summation over
$m=j-l$.
\end{proof}	
	
\begin{lemma}\label{lemma:constant fv}
Let $\fv a:\,[0,2\pi)\to[0,1]$ be a measurable function.
Assume that $\fv a$ is constant $c$ on some interval $[\lep,\uep]\subseteq[0,2\pi)$. 
Then for every $0<\delta<(\uep-\lep)/2$, and every 
$x\in[\lep+\delta,\uep-\delta]$, 
\begin{align*}
\abs{(\what S_n\fv a)(x)-c}\le\frac{\gamma_{\delta}}{n}\,,
\end{align*}
where $\gamma_{\delta}:=\frac{1}{\sin^2\frac{\delta}{2}}$.
\end{lemma}	
\begin{proof}
We may extend $\fv a$ periodically to $\bR$. Then 
\begin{align*}
(\what S_n\fv a)(x)-c
&=				
(\fv a\conv F_n)(x)-c 
= 
\int\limits_{-\pi}^{\pi}F_n(y)\left[a(x-y)-c)\right]\,\dd y
\end{align*}
for every $x$. 
If $x\in[\lep+\delta,\uep-\delta]$ then 
\begin{align*}
\abs{(\fv a\conv F_n)(x)-c} 
&\le 
\int\limits_{|y|<\delta}F_n(y)\underbrace{|\fv a(x-y)-c|}_{=0}\,\dd y
+\int\limits_{\delta\leq|y|\leq\pi}F_n(y)\underbrace{|\fv a(x-y)-c|}_{\le 1}\,\dd y \\
&\le 
\frac{1}{2\pi n}\int\limits_{\delta\leq|y|\leq\pi}\frac{\sin^2\frac{ny}{2}}{\sin^2\frac{y}{2}}\,\dd y
\le 
\frac{1}{2\pi n}\int\limits_{\delta\leq|y|\leq\pi}\frac{1}{\sin^2\frac{\delta}{2}}\,\dd y\\ 
&= 
\frac{\pi-\delta}{\pi n\sin^2\frac{\delta}{2}}\le \frac{1}{n\sin^2\frac{\delta}{2}}\,.
\end{align*}
\end{proof} 
	
\begin{lemma}\label{lemma:final bounds}
Let $[\lep,\uep]\subseteq[0,2\pi)$ be an interval, let 
$0<\delta<(\uep-\lep)/2$, and for every $n\in\bN$, let 
\begin{align*}
E_{n,\delta}:=\sum_{k:\,\frac{2\pi k}{n}\in[\lep+\delta,\uep-\delta]}
\pr{\Ft_n^*\egy_{\{k\}}}.
\end{align*}
Then $E_{n,\delta}$ is a projection on $\bC^{\n}$ such that 
\begin{align}\label{rank lower bound}
\Tr E_{n,\delta}\ge\floor{\frac{\uep-\lep-2\delta}{2\pi}n},
\end{align}
and for every measurable function $\fv a:\,[0,2\pi)\to[0,1]$ 
that is constant $0$ on $[\lep,\uep]$, 
\begin{align}\label{Tr EA upper bound}
\Tr E_{n,\delta}A_n\le(\Tr E_{n,\delta})\frac{\gamma_{\delta}}{n},\ds\ds\ds n\in\bN.
\end{align}
\end{lemma}	
\begin{proof}
Since $\Ft_n$ is a unitary, $E_{n,\delta}$ is indeed a projection, and the lower bound in 
\eqref{rank lower bound} is obvious. 
For any $k$ such that $\frac{2\pi k}{n}\in[\lep+\delta,\uep-\delta]$,
\begin{align*}
\Tr \pr{\Ft_n^*\egy_{\{k\}}}A_n=
\inner{\egy_{\{k\}}}{\Ft_nV_n^*AV_n\Ft_n^*\egy_{\{k\}}}
=
(\what S_n\fv a)\bz\frac{2\pi k}{n}\jz
\le
\frac{\gamma_{\delta}}{n},
\end{align*}
where the second equality is due to Lemma \ref{lemma:DFT diagonal}, and the inequality 
follows from Lemma \ref{lemma:constant fv}.
This immediately yields \eqref{Tr EA upper bound}.
\end{proof}	
	
\begin{thm}\label{thm:superexp2}
Let $\Omega_Q$ and $\Omega_R$ be as in Theorem \ref{thm:superexp}.
Then there exists 
a positive constant $c$ and 
a sequence of 
even projections $T_n\in\spann\{a(\phi):\,\phi\in\hil_n\}$ such that 
\begin{align}\label{main res detailed}
-\frac{1}{n}\log\sup_{i\in\I}\omega_{Q^{(i)}}(I-T_n)\ge 
c\log n,\ds\ds\ds
-\frac{1}{n}\log\sup_{j\in\J}\omega_{R^{(j)}}(T_n)\ge 
c\log n,\ds\ds\ds
n\in\bN.
\end{align}
In particular, $\Omega_Q$ and $\Omega_R$ can be super-exponentially distinguished by 
even projective tests. 

If, moreover, $\fv q_i$ is not almost everywhere $0$ and $\fv r_j$ is not 
almost everywhere $1$ on $[0,2\pi)$ for every $i\in\I$ and $j\in\J$, 
then their local densities $\what\omega_{Q_n^{(i)}}$ and 
$\what\omega_{R_n^{(j)}}$, $n\in\bN$, are all invertible.
\end{thm}	
\begin{proof}
Let $T_n$ be as in the proof of Corollary \ref{cor:error bounds2} with $E_n:=E_{n,\delta}$,
$n\in\bN$, for some $\delta$ as in Lemma \ref{lemma:final bounds}.
The inequalities in \eqref{nonas lower bound1}--\eqref{nonas lower bound2}
combined with \eqref{rank lower bound} and \eqref{Tr EA upper bound} yield 
\eqref{main res detailed}

The assertion about the invertibility of the density operators follows immediately from Corollary \ref{cor:invertible density}.
\end{proof}

\begin{ex}
Consider the $XX$ model with local Hamiltonian on $\B(\bC^2)_{[1,n]}$
given by 
\begin{align*}
H_n:=\half\sum_{k=-n}^{n-1}\bz\sigma_{x,k}\sigma_{x,k+1}+\sigma_{y,k}\sigma_{y,k+1}\jz+h\sum_{k=-n}^n\sigma_{z,k},
\end{align*}
where $\sigma_{x,k}$ is the Pauli $x$ operator
$\begin{bmatrix} 0 & 1 \\ 1 & 0\end{bmatrix}$ at site $k$, etc. 
It is well-known that the ground state of this model in the thermodynamic limit is 
$\tilde\omega_{Q^{(h)}}$, where $\omega_{Q^{(h)}}$ is the translation-invariant
quasi-free state corresponding to 
$\fv q_h:=\fv r_h:=\egy_{[\arccos f(h),2\pi-\arccos f(h))}$,
where $f(t):=\max\{-1,\min\{t,1\}\}$.
(Here $\egy_B$ stands for the indicator function of the set $B$.)
See, e.g., \cite[Appendix C]{MosonyiPhD} for a detailed exposition.

Let $h_0<h_1$ be such that $h_1>-1$ or $h_0<1$, and consider 
$\Omega_{Q}:=\{\fv q_{h}:\, h\le h_0 \}$,
$\Omega_{R}:=\{\fv r_{h}:\,h_1\le h\}$. 
That is, the experimenter's task is to test whether the transverse magnetic field
is below $h_0$ or above $h_1$, by making measurements on a finite part of the 
chain. It is straightforward to verify that 
$\fv q_h$ is constant zero on $[\lep:=\arccos f(h_1),\uep:=\arccos f(h_0)]$
for every $h\le h_0$, while
$\fv r_h$ is constant one on $[\lep,\uep]$
for every $h\ge h_1$, and hence,
by Theorem \ref{thm:superexp},
the two hypotheses can be tested with super-exponentially 
decreasing error probabilities.
By Corollary \ref{cor:invertible density}, the local densities $\what\omega_{Q_n^{(h)}}$ 
are invertible for every $h_1\le h<1$, and 
the local densities $\what\omega_{R_n^{(h)}}$ 
are invertible for every $-1<h\le h_0$.

A variant of the above problem is when the 
experimenter's task is to test whether the transverse magnetic field
is between $h_0$ and $h_0'$  or between $h_1$ and $h_1'$, where 
$-1<h_0<h_0'<h_1<h_1' <1$.
In this case
$\Omega_{Q}:=\{\fv q_{h}:\,h_0\le h\le h_0' \}$,
$\Omega_{R}:=\{\fv r_{h}:\,h_1\le h\le h_1'\}$. 
It is straightforward to verify that this problem satisfies the conditions in 
Theorem \ref{thm:superexp} with $\mu=\arccos h_1$, $\nu=\arccos h_0'$, and therefore
the two hypotheses can be tested with super-exponentially 
decreasing error probabilities, and, moreover, all local densities are invertible for every 
size $n$.
\end{ex}

\section{Comments on orthogonality}
\label{sec:ort}

In this section we discuss some relations between three concepts:
a) the orthogonality of a pair of states, b) their super-exponential distinguishability, and 
c) certain distinguishability measures taking infinite value on the given pair.
We start with an overview of the well-known relations between these
for density operators on a finite-dimensional Hilbert space, and then 
discuss a possible extension to pairs of translation-invariant states on an infinite spin chain.

Let $\bT(\hil):=\{T\in\B(\hil):\,0\le T\le I\}$ denote the set of \ki{tests} on $\hil$. 
It is well-known \cite{Helstrom1976,Holevo78} and easy to see that for any two density operators $\rho,\sigma\in\S(\hil)$,
\begin{align}\label{ort0}
\min_{T\in\bT(\hil)}\{\underbrace{\Tr\rho(I-T)}_{=:\alpha(T)}+\underbrace{\Tr\sigma T}_{=:\beta(T)}\}=1-\half\norm{\rho-\sigma}_1,
\end{align} 
where $\norm{X}_1:=\Tr|X|$, $X\in\B(\hil)$, is the trace-norm. In particular, we have 
\begin{align}\label{ort1}
\exists\,T\in\bT(\hil):\s\alpha(T)=0=\beta(T)\ds\iff\ds
\hh(\rho\|\sigma):=-\log\bz 1-\half\norm{\rho-\sigma}_1\jz=+\infty
\ds\iff\ds
\rho\perp\sigma,
\end{align}
where the first condition means perfect distinguishability of $\rho$ and $\sigma$, 
in the second condition we use the convention $\log 0:=-\infty$,
and the orthogonality in the last condition might be formulated in a number of different 
ways, e.g., as the orthogonality of the supports.

Orthogonality may be equivalently captured by various quantum R\'enyi divergences. For 
instance, for any $\alpha\in(0,1)$ and $z\in(0,+\infty)$, let 
\begin{align*}
D_{\alpha,z}(\rho\|\sigma):=\frac{1}{\alpha-1}\log\Tr\bz\rho^{\frac{\alpha}{z}}
\sigma^{\frac{1-\alpha}{z}}\rho^{\frac{\alpha}{z}}\jz^z
\end{align*}
be the R\'enyi $(\alpha,z)$-divergence of $\rho$ and $\sigma$ \cite{AD,JOPP}. Then 
\begin{align}\label{ort alphaz}
\rho\perp\sigma\ds\iff\ds D_{\alpha,z}(\rho\|\sigma)=+\infty\ds\text{for some/all}\ds
(\alpha,z)\ds\text{pairs as above.}
\end{align}
Note that the case $\alpha=1/2$, $z=1$, expresses the orthogonality of the 
unit vectors $\rho^{1/2}$, $\sigma^{1/2}$ in the Hilbert-Schmidt inner product.
Furthermore, for a test $T\in\bT(\hil)$, let 
\begin{align}\label{testmap}
\T(X):=(\Tr XT)\pr{0}+(\Tr X(I-T))\pr{1},
\end{align} 
where $\{\ket{0},\ket{1}\}$ is any orthonormal system in some Hilbert space, and let 
\begin{align*}
D_{\alpha}^{\test}(\rho\|\sigma):=\max_{T\in\bT(\hil)}D_{\alpha}(\T(\rho)\|\T(\sigma))
\end{align*}    
be the \ki{test-measured R\'enyi $\alpha$-divergence} of $\rho$ and $\sigma$ \cite{MH-testdiv}. In the 
above, 
\begin{align*}
D_{\alpha}(\T(\rho)\|\T(\sigma))&=D_{\alpha,z}(\T(\rho)\|\T(\sigma))\\
&=
\frac{1}{\alpha-1}\log\bz(\Tr \rho T)^{\alpha}(\Tr \sigma T)^{1-\alpha}
+
(\Tr \rho (I-T))^{\alpha}(\Tr \sigma (I-T))^{1-\alpha}\jz,\ds\ds\ds
z\in(0,+\infty),
\end{align*} 
is the classical R\'enyi divergence \cite{Renyi} of the commuting pair $\T(\rho),\T(\sigma)$. Then
\begin{align}\label{ort testmeas}
\rho\perp\sigma\ds\iff\ds D_{\alpha}^{\test}(\rho\|\sigma)=+\infty\ds\text{for some/all}\ds
\alpha\in(0,1),
\end{align}
which is just a reformulation of the equivalence of the first and the last conditions in \eqref{ort1}.

Consider now translation-invariant states $\omega^{(0)},\omega^{(1)}$ 
on the infinite spin chain algebra 
$\B(\hil)_{\bZ}$; see Section \ref{sec:state disc}.
One might define many of the above quantities directly for the states
$\omega^{(0)},\omega^{(1)}$. This is obvious for  $\chi(\omega^{(0)}\|\omega^{(1)})$ and $D_{\alpha}^{\test}(\omega^{(0)}\|\omega^{(1)})$; for the Petz-type R\'enyi divergences $D_{\alpha,1}\bz\omega^{(0)}\|\omega^{(1)}\jz$ and the sandwiched 
R\'enyi divergences $D_{\alpha,\alpha}\bz\omega^{(0)}\|\omega^{(1)}\jz$ with $\alpha\in[1/2,1)$, see, e.g., 
\cite{Hiai_fdiv_Springer,HM_sc_opalg,OP}. Most of these quantities, however, 
behave in a singular way for translation-invariant product states. Indeed, let 
$\omega_k$ denote the single-site density operator of $\omega^{(k)}$.
Additivity and monotonicity under restriction to subalgebras then gives
\begin{align*}
D_{\alpha,z}\bz\omega^{(0)}\|\omega^{(1)}\jz
&\ge
D_{\alpha,z}\bz\omega_0^{\otimes n}\|\omega_1^{\otimes n}\jz
=
nD_{\alpha,z}\bz\omega_0\|\omega_1\jz
\xrightarrow[n\to+\infty]{}+\infty
\end{align*}    
for every $\alpha\in(0,1)$ and $z=1$, or $z=\alpha\in[1/2,1)$, whenever $\omega_0\ne\omega_1$.
Using the Fuchs-van de Graaf inequality \cite{FvdG} in the form
$\hh(\rho\|\sigma)\ge \half D_{1/2,1/2}(\rho\|\sigma)$ then yields
\begin{align*}
\hh(\omega^{(0)}\|\omega^{(1)})\ge\hh(\omega_0^{\otimes n}\|\omega_1^{\otimes n})
\ge
\half D_{1/2,1/2}(\omega_0^{\otimes n}\|\omega_1^{\otimes n})\xrightarrow[n\to+\infty]{}+\infty.
\end{align*}
In fact, it is also known that $\hh(\omega^{(0)}\|\omega^{(1)})=+\infty$ whenever 
$\omega^{(0)}$ and $\omega^{(1)}$ are different ergodic states 
\cite[Corollary IV.4.2]{Israel}, and it is not too difficult to see that 
translation-invariant quasi-free states are ergodic 
(for a hint, see, e.g., \cite[Example 7.6]{AlickiFannes-book}).

In view of the above, the above considered distinguishability measures
defined directly for the infinite spin chain states
reveal very little about the relation of 
$\omega^{(0)}$ and $\omega^{(1)}$ from the point of view of state discrimination.
One might consider instead the regularized versions of the above quantities, defined for 
two translation-invariant states $\omega^{(0)}$ and $\omega^{(1)}$ as
\begin{align}\label{regmeas}
\oll\divv(\omega^{(0)}\|\omega^{(1)}):=\liminf_{n\to+\infty}\frac{1}{n}\divv\bz\omega^{(0)}_{[1,n]}\big\|\omega^{(1)}_{[1,n]}\jz,
\end{align}
where $\divv$ may stand for any distinguishability measure on pairs of states, like
$\hh$, $D_{\alpha,z}$, $D_{\alpha}^{\test}$, etc. 

\begin{thm}\label{thm:ort}
Let $\omega^{(0)}$ and $\omega^{(1)}$ be translation-invariant states on the infinite 
spin-chain algebra $\B(\hil)_{\bZ}$. The following are equivalent:
\begin{enumerate}
\item\label{proport1}
$\omega^{(0)}$ and $\omega^{(1)}$ can be super-exponentially distinguished. 
\item\label{proport2}
$\oll\hh\bz\omega^{(0)}\big\|\omega^{(1)}\jz=+\infty$.
\item\label{proport3}
$\oll D_{\alpha}^{\test}\bz\omega^{(0)}\big\|\omega^{(1)}\jz=+\infty$ for every
$\alpha\in(0,1)$.
\item\label{proport4}
$\oll D_{\alpha}^{\test}\bz\omega^{(0)}\big\|\omega^{(1)}\jz=+\infty$ for some 
$\alpha\in(0,1)$.
\item\label{proport5}
$\oll D_{\alpha,z}\bz\omega^{(0)}\big\|\omega^{(1)}\jz=+\infty$ for every
$\alpha\in(0,1)$ and every $z\ge \max\{\alpha,1-\alpha\}$.
\item\label{proport6}
$\oll D_{\alpha,z}\bz\omega^{(0)}\big\|\omega^{(1)}\jz=+\infty$ for some
$\alpha\in(0,1)$ and some $z\ge\max\{\alpha,1-\alpha\}$.
\end{enumerate}
\end{thm}    
\begin{proof}
The equivalence \ref{proport1}$\iff$\ref{proport2} is clear from \eqref{ort0}.
It is straightforward to verify that 
\ref{proport1} yields 
$\oll D_{\alpha}^{\test}\bz\omega^{(0)}\big\|\omega^{(1)}\jz=+\infty$ for every
$\alpha\in(0,1)$, proving \ref{proport1}$\imp$\ref{proport3}. 
The implication \ref{proport3}$\imp$\ref{proport4} is obvious. 

Assume \ref{proport4}, i.e., that 
$\oll D_{\alpha}^{\test}\bz\omega^{(0)}\big\|\omega^{(1)}\jz=+\infty$ for some
$\alpha\in(0,1)$.
Then there exists a test sequence $T_n\in\B(\hil)_{[1,n]}$, $n\in\bN$, 
and a sequence $c_n\in[0,+\infty)$, $n\in\bN$, with $\lim_n c_n=+\infty$,
such that 
\begin{align*}
e^{-(1-\alpha)nc_n}
&=
(\Tr \omega^{(0)}_{n} T_n)^{\alpha}(\Tr \omega^{(1)}_{n} T_n)^{1-\alpha}
+
(\Tr\omega^{(0)}_{n} (I-T_n))^{\alpha}(\Tr \omega^{(1)}_{n} (I-T_n))^{1-\alpha}\\
&\ge
\min\left\{(\Tr \omega^{(1)}_{n} T_n)^{1-\alpha},(\Tr \omega^{(1)}_{n}(I-T_n))^{1-\alpha}\right\}
\underbrace{\bz (\Tr \omega^{(0)}_{n} T_n)^{\alpha}+(\Tr \omega^{(0)}_{n} (I-T_n))^{\alpha}\jz}_{\ge 1},
\end{align*}
where $\omega^{(k)}_n:=\omega^{(k)}_{[1,n]}$. Let us define a new test sequence
$\tilde T_n:=T_n$ if $\Tr\omega^{(1)}_{n} T_n\le 1/2$, and 
$\tilde T_n:=I-T_n$ otherwise. Then the above yields
$\Tr \omega^{(1)}_{n} \tilde T_n\le e^{-nc_n}$, which goes to zero super-exponentially, and
\begin{align*}
e^{-(1-\alpha)nc_n}
&\ge
(\Tr\omega^{(0)}_{n} (I-\tilde T_n))^{\alpha}(\Tr \omega^{(1)}_{n} (I-\tilde T_n))^{1-\alpha}
\ge
(1-e^{-nc_n})^{1-\alpha}\Tr\omega^{(0)}_{n} (I-\tilde T_n))^{\alpha},
\end{align*}
whence $\Tr\omega^{(0)}_{n} (I-\tilde T_n)$ also goes to zero super-exponentially in $n$.
Thus, we obtain \ref{proport1}.

According to \cite[Theorem 1.1]{Hi3} and a standard argument deriving monotonicity under CPTP 
maps from joint convexity, the R\'enyi $(\alpha,z)$-divergences are monotone non-increasing 
under the joint action of a CPTP map on both of their arguments when 
$\alpha\in(0,1)$ and $\max\{\alpha,1-\alpha\}\le z$.
This immediately implies
$\oll D_{\alpha,z}\bz\omega^{(0)}\big\|\omega^{(1)}\jz
\ge
\oll D_{\alpha}^{\test}\bz\omega^{(0)}\big\|\omega^{(1)}\jz$
for any such $\alpha,z$, and thus the implication 
\ref{proport3}$\imp$\ref{proport5} follows, and
\ref{proport5}$\imp$\ref{proport6} is trivial.

Finally, \ref{proport6}$\imp$\ref{proport1}
follows immediately from Corollary \ref{cor:Aud2}.
\end{proof}    

\begin{rem}
It is clear from the proof of Theorem \ref{thm:ort} that the following also holds.
If $n_1<n_2<\ldots$, and $\divv=\hh$, $\divv=D_{\alpha,z}$ with $\alpha\in(0,1)$ and 
$z\ge\max\{\alpha,1-\alpha\}$, or $\divv=D_{\alpha}^{\test}$ with $\alpha\in(0,1)$, then the following are equivalent:
\begin{enumerate}
\item\label{subsequence1}
There exists a sequence of tests $T_{n_k}\in\B(\hil)_{[1,n_k]}$, $k\in\bN$, such that 
\begin{align*}
\lim_{k\to+\infty}-\frac{1}{n_k}\log\Tr\omega^{(0)}_{[1,n_k]}(I-T_{n_k})=+\infty
=
\lim_{k\to+\infty}-\frac{1}{n_k}\log\Tr\omega^{(1)}_{[1,n_k]}T_{n_k}.
\end{align*}
\item\label{subsequence2}
$\displaystyle{\lim_{k\to+\infty}\frac{1}{n_k}\divv\bz
\omega^{(0)}_{[1,n_k]}\big\|\omega^{(1)}_{[1,n_k]}\jz=+\infty}$.
\end{enumerate}

Indeed, \ref{subsequence1} is equivalent to \ref{subsequence2}
with $\divv=\hh$ due to \eqref{ort0}, which implies 
\ref{subsequence2} with $\divv=D_{\alpha}^{\test}$ for any given $\alpha\in(0,1)$;
this implies \ref{subsequence2} with $\divv=D_{\alpha,z}$ for the same $\alpha$
and any $z\ge\max\{\alpha,1-\alpha\}$, due to monotonicity under CPTP maps;
and finally this implies \ref{subsequence2} with $\divv=\hh$
due to Corollary \ref{cor:Aud2}.  

In particular, Theorem \ref{thm:ort} remains valid if we replace the $\liminf$
with $\limsup$ in the definition of the error exponents 
in \eqref{errexp}, and define 
super-exponential distinguishability accordingly, and we also 
replace the $\liminf$ with $\limsup$ in
the definition of the regularized distinguishability measures in 
\eqref{regmeas}. 
\end{rem}

\begin{rem}
Two states (positive linear normalized functionals)
$\omega^{(0)}$ and $\omega^{(1)}$ on a $C^*$-algebra 
are defined to be orthogonal in \cite[Definition 1.14.1]{Sakai-book} if 
$\norm{\omega^{(0)}-\omega^{(1)}}=2$, where the norm is the usual functional norm; this is equivalent to $\hh(\omega^{(0)}\big\|\omega^{(1)})=+\infty$ in our notation.
The above arguments show that this notion of orthogonality may not be the best suited for the study of asymptotic state discrimination on an infinite spin chain; in particular, 
any two translation-invariant product states are orthogonal according to this definition, irrespective of whether the density operators of their local restrictions are orthogonal or not. 

In contrast, if we define $\omega^{(0)}$ and $\omega^{(1)}$ on an infinite spin chain to 
be orthogonal if
$\oll\hh(\omega^{(0)}\big\|\omega^{(1)})=+\infty$ then for 
translation-invariant product states
this becomes equivalent to the usual orthogonality of their single-site restrictions
$\omega^{(0)}_{[1]}$ and $\omega^{(1)}_{[1]}$. Another appealing feature of this 
notion of orthogonality of states is that it is equivalent to various regularized
distinguishability measures being $+\infty$, according to Theorem \ref{thm:ort}, 
which gives a nice generalization of the analogous single-site characterizations
of orthogonality given in \eqref{ort alphaz} and \eqref{ort testmeas}.
Of course, this notion of orthogonality is limited to pairs of translation-invariant states
on an infinite spin chain, and does not make sense in general for pairs of states on an abstract $C^*$-algebra.
\end{rem}

\begin{rem}
Clearly, if any (and hence all) of \ref{proport1}--\ref{proport6} in Theorem 
\ref{thm:ort} holds then we have 
$\oll{\mathbb{D}}_{\alpha}(\omega^{(0)}\big\|\omega^{(1)})=+\infty$ for any quantum R\'enyi 
$\alpha$-divergence with $\alpha\in(0,1)$ that is monotone non-increasing under 
$2$-outcome measurements, i.e., under the type of CPTP maps given in \eqref{testmap}.
Here, we say that $\mathbb{D}_{\alpha}$ is a quantum R\' enyi $\alpha$-divergence if it 
is defined on all pairs of density operators on any finite-dimensional Hilbert space, and for commuting states it reduces to the classical R\'enyi $\alpha$-divergence 
of the diagonal elements of the two density operators in a common eigen-basis.
One such example is Matsumoto's maximal $\alpha$-divergence \cite{Matsumoto_newfdiv}
for every $\alpha\in(0,1)$; however, at the moment we do not know 
if the regularized maximal $\alpha$-divergence being $+\infty$ implies the other 
properties listed in Theorem \ref{thm:ort}.
\end{rem}

\section{Conclusion}

We have shown that translation-invariant quasi-free states with defining functions
$\fv q$ and $\fv r$ are super-exponentially distinguishable if there is an interval
$[\lep,\uep]$ of non-zero length such that one of the functions is constant $0$
and the other one is constant $1$ on this interval. We have shown that in this case 
both errors decreases at least as fast as $e^{-nc\log n}$ in the sample size $n$; it is 
however, an open question whether this is in fact the optimal asymptotics, or a faster 
decrease, e.g., $e^{-cn^{1+\delta}}$ with some $\delta>0$ can be attained.
This can be asked for the class of functions that we considered, but it is also natural to 
ask if there is any upper bound on the speed of convergence to zero for 
general pairs of translation-invariant states on a spin chain.

It is known that a translation-invariant quasi-free state $\omega_Q$ 
is pure (i.e., an extremal point of the convex set of states) if and only if the 
corresponding function $\fv q$ is an indicator function, 
i.e., $\fv q=\egy_{B_Q}$ for some measurable subset $B_Q$ of $[0,2\pi)$
(see, e.g., \cite{Fannes-CAR}).
Two such pure states $\omega_Q$ and $\omega_R$ are different if and only if 
$B_Q$ and $B_R$ are different in the measure-theoretic sense, i.e., the Lebesgue measure of
$(B_Q\setminus B_R)\cup(B_R\setminus B_Q)$ is positive. This motivates to ask whether the following extension of our result is true: If 
$\fv q$ and $\fv r$ are measurable functions from $[0,2\pi)$ to $[0,1]$
such that there exists a measurable set $B\subseteq[0,2\pi)$ of positive Lebesgue measure on which $\fv q$ is constant $0$ and $\fv r$ is constant $1$ then 
$\omega_Q$ and $\omega_R$ can be super-exponentially distinguished. 
In particular, this would imply the super-exponential distinguishability 
of any two different pure translation-invariant quasi-free states.
    
\section*{Acknowledgments}

This work was partially funded by the
National Research, Development and 
Innovation Office of Hungary via the research grants K 124152, KH 129601, and
FK 135220, and
by the Ministry of Innovation and
Technology and the National Research, Development and Innovation
Office within the Quantum Information National Laboratory of Hungary.

\appendix
\section{A variant of Audenaert's inequality}

It was shown in \cite[Theorem 1]{Hoeffding1} that for any two density operators 
$\rho,\sigma$ on a finite-dimensional Hilbert space $\hil$,
\begin{align*}
1-\half\norm{\rho-\sigma}_1\le\Tr\rho^{\alpha}\sigma^{1-\alpha},\ds\ds\ds\alpha\in(0,1),
\end{align*}
or equivalently,
\begin{align}\label{Aud}
\hh(\rho\|\sigma):=-\log\bz 1-\half\norm{\rho-\sigma}_1\jz
\ge(1-\alpha)D_{\alpha,1}(\rho\|\sigma).
\end{align}
(See Section \ref{sec:ort} for the definition of the R\'enyi 
$(\alpha,z)$-divergences.) We will need a simple extension of the above, given in Corollary \ref{cor:Aud2} below, which follows by a 
combination of \eqref{Aud}, the Araki-Lieb-Thirring (ALT) inequality \cite{Araki,LT}, 
and its 
converse given in \cite{Aud-ALT}.
Recall that the ALT inequality states that for any two positive semi-definite operators
$A,B\in\B(\hil)$,
\begin{align}\label{ALT}
\Tr\bz A^rB^rA^r\jz^q\le\Tr\bz ABA\jz^{rq},\ds\ds\ds
q\in[0,+\infty),\ds r\in[0,1].
\end{align}
The converse given in \cite[Theorem 2]{Aud-ALT} states that 
\begin{align}\label{CALT}
\Tr(ABA)^{rq}\le\bz\Tr\bz A^rB^rA^r\jz^q\jz^r\bz\norm{A}^{2rq}\Tr B^{rq}\jz^{1-r},
\ds\ds\ds
q\in[0,+\infty),\ds r\in[0,1].
\end{align}

\begin{lemma}
Let $\rho,\sigma$ be density operators on a finite-dimensional Hilbert space $\hil$.
For any $\alpha\in(0,1)$,
\begin{align}\label{alpha-z bounds}
D_{\alpha,1}(\rho\|\sigma)\ge
\begin{cases}
D_{\alpha,z}(\rho\|\sigma),&z\in(0,1],\\
\frac{1}{z}D_{\alpha,z}(\rho\|\sigma)-\frac{\alpha}{(1-\alpha)^2}\frac{z-1}{z}\log\dim\hil,&z>1.
\end{cases}
\end{align}
\end{lemma}
\begin{proof}
Let $z'\le z''$,
$A:=\rho^{\frac{\alpha}{2z'}}$, 
$B:=\sigma^{\frac{1-\alpha}{z'}}$, $q:=z''$, $r:=z'/z''$.
Then \eqref{ALT} gives 
\begin{align*}
\Tr\bz\rho^{\frac{\alpha}{z''}}
\sigma^{\frac{1-\alpha}{z''}}\rho^{\frac{\alpha}{z''}}\jz^{z''}
\le
\Tr\bz\rho^{\frac{\alpha}{z'}}
\sigma^{\frac{1-\alpha}{z'}}\rho^{\frac{\alpha}{z'}}\jz^{z'},
\end{align*}
or equivalently, $D_{\alpha,z''}(\rho\|\sigma)\ge D_{\alpha,z'}(\rho\|\sigma)$.
In particular, the choice $z'':=1$, $z':=z$ yields the first inequality in \eqref{alpha-z bounds}.

Now, let $z':=1$, $z'':=z>1$. Then \eqref{CALT} gives
\begin{align}\label{alpha-z bounds proof}
\Tr\rho^{\alpha}\sigma^{1-\alpha}\le
\bz\Tr\bz\rho^{\frac{\alpha}{z}}
\sigma^{\frac{1-\alpha}{z}}\rho^{\frac{\alpha}{z}}\jz^{z}\jz^{1/z}
\underbrace{\bz\norm{\rho}^{\alpha}\jz^{1-1/z}}_{\le 1}
\bz\Tr\sigma^{1-\alpha}\jz^{1-1/z}.
\end{align}
Let $\lambda_1,\ldots,\lambda_d$ be the eigen-values of $\sigma$, where $d:=\dim\hil$. 
Since $1/(1-\alpha)>0$, we get 
\begin{align*}
(\Tr\sigma^{1-\alpha})^{\frac{1}{1-\alpha}}
=
\bz\sum_{i=1}^d\lambda_i^{1-\alpha}\jz^{\frac{1}{1-\alpha}}
=
\bz d\sum_{i=1}^d\frac{1}{d}\lambda_i^{1-\alpha}\jz^{\frac{1}{1-\alpha}}
\le
d^{\frac{1}{1-\alpha}}\sum_{i=1}^d\frac{1}{d}\lambda_i=
d^{\frac{\alpha}{1-\alpha}}.
\end{align*}
Writing this back into \eqref{alpha-z bounds proof} we get 
\begin{align*}
\underbrace{\frac{1}{\alpha-1}\log \Tr\rho^{\alpha}\sigma^{1-\alpha}}_{=D_{\alpha,1}(\rho\|\sigma)}
\ge
\frac{1}{z}\underbrace{\frac{1}{\alpha-1}\log\Tr\bz\rho^{\frac{\alpha}{z}}
\sigma^{\frac{1-\alpha}{z}}\rho^{\frac{\alpha}{z}}\jz^{z}}_{=
\frac{1}{z}D_{\alpha,z}(\rho\|\sigma)}
-\frac{\alpha}{(1-\alpha)^2}\frac{z-1}{z}\log d,
\end{align*}
which is exactly the second inequality in \eqref{alpha-z bounds}.
\end{proof}

\begin{cor}\label{cor:Aud2}
Let $\rho,\sigma$ be density operators on a finite-dimensional Hilbert space $\hil$.
For every $\alpha\in(0,1)$,
\begin{align*}
\hh(\rho\|\sigma)
\ge 
\begin{cases}
(1-\alpha)D_{\alpha,z}(\rho\|\sigma),&z\in(0,1],\\
\frac{1}{z}(1-\alpha)D_{\alpha,z}(\rho\|\sigma)-\frac{\alpha}{1-\alpha}\frac{z-1}{z}\log\dim\hil,&z>1.
\end{cases}
\end{align*}
\end{cor}
\begin{proof}
Immediate from \eqref{Aud} and \eqref{alpha-z bounds}.
\end{proof}

\bibliography{bibliography220318}

\end{document}